\documentclass[graybox]{svmult}

\usepackage{mathptmx}       
\usepackage{helvet}         
\usepackage{courier}        
\usepackage{type1cm}        
\usepackage{makeidx}         
\usepackage{graphicx}        
\usepackage{multicol}        
\usepackage[bottom]{footmisc}
\usepackage{amssymb}
\usepackage{amsfonts,amsmath}

\makeindex             

\begin{document}

\title*{Chiral properties of discrete Joyce and Hestenes  equations}
\author{Volodymyr Sushch}
\institute{Volodymyr Sushch \at Koszalin University of Technology, Sniadeckich 2,
 75-453 Koszalin, Poland \at \email{volodymyr.sushch@tu.koszalin.pl}}
%
%
\maketitle

\abstract*{This paper concerns the question of how chirality is realized for discrete counterparts of the Dirac-K\"{a}hler equation in the Hestenes and Joyce forms. It is shown that left and right chiral states for these discrete equations can be described with the aid of some projectors on a space of discrete forms. The proposed discrete model admits a chiral symmetry. We construct discrete analogs of spin operators, describe spin eigenstates for a discrete Joyce equation, and also discuss chirality.}

\abstract{This paper concerns the question of how chirality is realized for discrete counterparts of the Dirac-K\"{a}hler equation in the Hestenes and Joyce forms. It is shown that left and right chiral states for these discrete equations can be described with the aid of some projectors on a space of discrete forms. The proposed discrete model admits a chiral symmetry. We construct discrete analogues of spin operators, describe spin eigenstates for a discrete Joyce equation, and also discuss chirality.}



\section{Introduction}
\label{sec:1}
We present some recent results in the discretisation of the   Dirac  equation in the geometric algebra of spacetime by using the  Dirac-K\"{a}hler approach.
In this approach, a discretisation scheme is geometric in nature and rests upon the use of the differential forms calculus. The general topic of this paper is the description of some discrete constructions in which the chiral properties of the Dirac theory are captured. In the context of the geometric discretisation, it is natural to introduce a Clifford product acting on the space of discrete inhomogeneous forms as was discussed in \cite{S6}.
  This work is a direct continuation of  that described in my previous papers \cite{S1, S2, S3, S4, S5, S6}. In \cite{S2}, on the issue of chirality, special attention to a discrete Hodge star operator has been paid. A central role of the Hodge star to deal with chiral symmetry in the lattice formulation  was already pointed out by Rabin \cite{Rabin}.
 There are several approaches to study of discrete versions of the Dirac-K\"{a}hler equation based on the use of a discrete Clifford calculus framework on lattices. For a review of discrete Clifford analysis, we refer the reader to \cite{FKS, F2, F3, Kanamori, Vaz}.

We first briefly review some notations and basic facts  on the Dirac-K\"{a}hler equation \cite{Kahler, Rabin} and the Dirac equation
 in the spacetime algebra \cite{H1, H2}.
Let $M={\mathbb R}^{1,3}$ be  Minkowski space.
Denote by $\Lambda^r(M)$ the vector space of smooth complex-valued differential $r$-forms, $r=0,1,2,3,4$. 
 Let $d:\Lambda^r(M)\rightarrow\Lambda^{r+1}(M)$ be the exterior differential and let $\delta:\Lambda^r(M)\rightarrow\Lambda^{r-1}(M)$ be the formal adjoint of $d$  with respect to  natural inner  product in $\Lambda^r(M)$. We have $$\delta=\ast d\ast,$$
where $\ast$ is the Hodge star operator  $\ast:\Lambda^r(M)\rightarrow\Lambda^{4-r}(M)$ with respect to the Lorentz metric.
 Denote by $\Lambda(M)$ the set of all differential forms on $M$. We have
\begin{equation*}
\Lambda(M)=\Lambda^0(M)\oplus\Lambda^1(M)\oplus\Lambda^2(M)\oplus\Lambda^3(M)\oplus\Lambda^4(M).
\end{equation*}
Let $\Omega\in\Lambda(M)$
be an inhomogeneous differential form, i.e.,
$\Omega=\sum_{r=0}^4\overset{r}{\omega}$,
where $\overset{r}{\omega}\in\Lambda^r(M)$.
 The Dirac-K\"{a}hler equation for a free electron is given by
\begin{equation}\label{eq:01}
i(d+\delta)\Omega=m\Omega,
\end{equation}
where $i$ is the usual complex unit  and  $m$  is a mass parameter.

Let  $\{\gamma_0, \gamma_1, \gamma_2, \gamma_3\}$ be  a vector basis of the Clifford algebra $\emph{C}\ell(1,3)$, namely $\gamma_\mu\gamma_\nu+\gamma_\nu\gamma_\mu=g_{\mu\nu},$
where $g_{\mu\nu}=\mbox{diag}(1,-1,-1,-1)$ and $\mu, \nu=0,1,2,3$.
Hestenes \cite{H2} calls  this algebra the spacetime algebra.
It is known that the vectors $\gamma_\mu$ can be represented by the $4\times 4$ Dirac gamma matrices
 \cite{B, H2}.
Through the identification of the basic covectors $dx^\mu$ and the matrices $\gamma_\mu$ which arises from representation theory, one connects the differential forms under the Clifford product to the algebra of gamma matrices.
In other words, the graded algebra $\Lambda(M)$ endowed with the Clifford multiplication is an example of a Clifford algebra. It is true that Eq.~\eqref{eq:01} is equivalent to the four usual Dirac equations (traditional column-spinor equations).
Let $\Lambda_{\mathbb{R}}(M)$ denote the set of real-valued differential forms and let $\Lambda^{ev}(M)=\Lambda^0(M)\oplus\Lambda^2(M)\oplus\Lambda^4(M)$. The Dirac equation in the Hestenes form \cite{H1, H2} can be written in terms of inhomogeneous forms as

\begin{equation}\label{eq:02}
-(d+\delta)\Omega^{ev} \gamma_1\gamma_2=m\Omega^{ev}\gamma_0, \quad \Omega^{ev}\in\Lambda_{\mathbb{R}}^{ev}(M).
 \end{equation}
 We consider also the generalized bivector Dirac equation  \cite{J} in the form
 \begin{equation}\label{eq:03}
 i(d+\delta)\Omega^{ev}=m\Omega^{ev}\gamma_0, \quad \Omega^{ev}\in\Lambda^{ev}(M).
 \end{equation}
Following Baylis \cite{B1} we call Eq.~(\ref{eq:03}) the Joyce equation. This equation is equivalent to two copies of the usual Dirac equation. For a deeper discussion of equivalence of Dirac formulations we refer the reader to \cite{JM}.

The goal of this work is to establish the chirality of discrete versions of the Dirac equation in the  Hestenes and Joyce forms.    We show that defined some projectors on the space of discrete forms one can decompose  solutions of Eqs.~(\ref{eq:01})--(\ref{eq:03})  into its left-handed and right-handed parts. Two types of such projectors are introduced  and we prove that a discrete Dirac-K\"{a}hler operator flips the chirality for both of them.
We also construct spin $\pm \frac{1}{2}$ eigenstates for a discrete counterpart of the plane wave solution to a discrete Joyce equation and  discuss chirality for such fields.

\section{Discrete Dirac-K\"{a}hler, Hestenes and Joyce equations}
\label{sec:2}
In this section, we recall some discrete constructions concerning the  Dirac-K\"{a}hler equation and a discrete Clifford calculus.   A discretization scheme  is based on the language of differential forms and  is described in  \cite{S2}. The approach was originated by Dezin  \cite{Dezin}. For the convenience of the reader we briefly repeat the relevant material from \cite{S2}
without proofs, thus making our presentation self-contained.   All details one can  find in   \cite{S1, S2}.

Let $K(4)=K\otimes K\otimes K\otimes K$
be a cochain complex with  complex  coefficients,
where  $K$ is  the 1-dimensional complex generated by 0- and 1-dimensional basis elements   $x^{k_\mu}$  and $e^{k_\mu}$,  $k_\mu\in\mathbb{Z}$,  respectively.
Then an arbitrary $r$-dimensional basis element of $K(4)$ can be written as
$s^k_{(r)}=s^{k_0}\otimes s^{k_1}\otimes s^{k_2}\otimes s^{k_3}$, where
$s^{k_\mu}$ is either $x^{k_\mu}$ or $e^{k_\mu}$,\ $\mu=0,1,2,3$ and  $k=(k_0,k_1,k_2,k_3)$ is a multi-index. The symbol $(r)$ contains the whole required information about the number and position  $e^{k_\mu}\in K$ in
$s^k_{(r)}\in K(4)$.
For example, the 1-dimensional basis elements of
$K(4)$ can be written as
\begin{eqnarray*}
e^k_0=e^{k_0}\otimes x^{k_1}\otimes x^{k_2}\otimes x^{k_3},  \qquad
e^k_1=x^{k_0}\otimes e^{k_1}\otimes x^{k_2}\otimes x^{k_3}, \\
e^k_2=x^{k_0}\otimes x^{k_1}\otimes e^{k_2}\otimes x^{k_3},  \qquad
e^k_3=x^{k_0}\otimes x^{k_1}\otimes x^{k_2}\otimes e^{k_3}.
\end{eqnarray*}
The 2-dimensional basis elements of $K(4)$  have the form
\begin{eqnarray*}
e^k_{01}=e^{k_0}\otimes e^{k_1}\otimes x^{k_2}\otimes x^{k_3}, \ \  e^k_{02}=e^{k_0}\otimes x^{k_1}\otimes e^{k_2}\otimes x^{k_3}, \ \ 
e^k_{03}=e^{k_0}\otimes x^{k_1}\otimes x^{k_2}\otimes e^{k_3}, \\
e^k_{12}=x^{k_0}\otimes e^{k_1}\otimes e^{k_2}\otimes x^{k_3},  \ \  e^k_{13}=x^{k_0}\otimes e^{k_1}\otimes x^{k_2}\otimes e^{k_3}, \ \ 
e^k_{23}=x^{k_0}\otimes x^{k_1}\otimes e^{k_2}\otimes e^{k_3}.
\end{eqnarray*}
In the same way one can write down the 3-dimensional basic elements $e^k_{012}$, $e^k_{013}$, $e^k_{023}$ and $e^k_{123}$.
Finally, denote by
\begin{equation*}
x^k=x^{k_0}\otimes x^{k_1}\otimes x^{k_2}\otimes x^{k_3}, \qquad e^k=e^{k_0}\otimes e^{k_1}\otimes e^{k_2}\otimes e^{k_3}
\end{equation*}
 the 0- and 4-dimensional basis elements of $K(4)$.

The complex $K(4)$ is a discrete analogue of $\Lambda(M)$ and cochains play a role of differential
forms. Let us call them forms or discrete forms to emphasize their relationship with differential
forms.    Then we have
\begin{equation*}
K(4)=K^0(4)\oplus K^1(4)\oplus K^2(4)\oplus K^3(4)\oplus K^4(4),
\end{equation*}
where $K^r(4)$ denotes the set of all discrete $r$-forms, and
  any  $\overset{r}{\omega}\in K^r(4)$ can be expressed as
\begin{eqnarray}\label{eq:04}
\overset{0}{\omega}=\sum_k\overset{0}{\omega}_kx^k, \qquad \overset{2}{\omega}=\sum_k\sum_{\mu<\nu} \omega_k^{\mu\nu}e_{\mu\nu}^k,  \qquad  \overset{4}{\omega}=\sum_k\overset{4}{\omega}_ke^k, \\ \label{eq:05}
\overset{1}{\omega}=\sum_k\sum_{\mu=0}^3\omega_k^\mu e_\mu^k, \qquad
\overset{3}{\omega}=\sum_k\sum_{\iota<\mu<\nu} \omega_k^{\iota\mu\nu}e_{\iota\mu\nu}^k,
\end{eqnarray}
where  $\overset{0}{\omega}_k, \ \omega_k^{\mu\nu}, \ \overset{4}{\omega}_k, \  \omega_k^\mu$ and $\omega_k^{\iota\mu\nu}$ are complex numbers.

Let $d^c: K^r(4)\rightarrow K^{r+1}(4)$ be a discrete analogue of the exterior derivative $d$ and let $\delta ^c: K^r(4)\rightarrow K^{r-1}(4)$ be a discrete analogue of the codifferential $\delta$. It is clear that $\delta^c=\ast d^c\ast$. For more precise  definitions of these operators we refer the reader to  \cite{S2}. In this paper  we give only the difference
expressions for $d^c$ and  $\delta^c$.
Let the difference operator $\Delta_\mu$ be defined by
\begin{equation}\label{eq:06}
\Delta_\mu\omega_k^{(r)}=\omega_{\tau_\mu k}^{(r)}-\omega_k^{(r)},
\end{equation}
where  $\omega_k^{(r)}\in\mathbb{C}$ is a component of $\overset{r}{\omega}\in K^r(4)$ and
$\tau_\mu$ is   the shift operator  which acts  as
$\tau_\mu k=(k_0,...k_\mu+1,...k_3), \   \mu=0,1,2,3.$
For forms (\ref{eq:04}), (\ref{eq:05})  we have
\begin{eqnarray}\label{eq:07}
d^c\overset{0}{\omega}=\sum_k\sum_{\mu=0}^3(\Delta_\mu\overset{0}{\omega}_k)e_\mu^k,  \qquad d^c\overset{1}{\omega}=\sum_k\sum_{\mu<\nu}(\Delta_\mu\omega_k^\nu-\Delta_\nu\omega_k^\mu)e_{\mu\nu}^k,
\end{eqnarray}
\begin{eqnarray}\label{eq:08}
d^c\overset{2}{\omega}=\sum_k\big[(\Delta_0\omega_k^{12}-\Delta_1\omega_k^{02}+\Delta_2\omega_k^{01})e_{012}^k
+(\Delta_0\omega_k^{13}-\Delta_1\omega_k^{03}+\Delta_3\omega_k^{01})e_{013}^k \nonumber \\
+(\Delta_0\omega_k^{23}-\Delta_2\omega_k^{03}+\Delta_3\omega_k^{02})e_{023}^k
+(\Delta_1\omega_k^{23}-\Delta_2\omega_k^{13}+\Delta_3\omega_k^{12})e_{123}^k\big],
\end{eqnarray}
\begin{equation}\label{eq:09}
d^c\overset{3}{\omega}=\sum_k(\Delta_0\omega_k^{123}-\Delta_1\omega_k^{023}+\Delta_2\omega_k^{013}-\Delta_3\omega_k^{012})e^k,  \quad \qquad d^c\overset{4}{\omega}=0,
\end{equation}
\begin{equation}\label{eq:10}
\delta^c\overset{0}{\omega}=0, \quad \qquad \delta^c\overset{1}{\omega}=\sum_k(\Delta_0\omega_k^{0}-\Delta_1\omega_k^{1}-\Delta_2\omega_k^{2}-\Delta_3\omega_k^{3})x^k,
\end{equation}
\begin{eqnarray}\label{eq:11} \nonumber
\delta^c\overset{2}{\omega}=\sum_k\big[(\Delta_1\omega_k^{01}+\Delta_2\omega_k^{02}+\Delta_3\omega_k^{03})e_{0}^k
+(\Delta_0\omega_k^{01}+\Delta_2\omega_k^{12}+\Delta_3\omega_k^{13})e_{1}^k\\
+(\Delta_0\omega_k^{02}-\Delta_1\omega_k^{12}+\Delta_3\omega_k^{23})e_{2}^k
+(\Delta_0\omega_k^{03}-\Delta_1\omega_k^{13}-\Delta_2\omega_k^{23})e_{3}^k\big],
\end{eqnarray}
\begin{eqnarray}\label{eq:12} \nonumber
\delta^c\overset{3}{\omega}=\sum_k\big[(-\Delta_2\omega_k^{012}-\Delta_3\omega_k^{013})e_{01}^k+
(\Delta_1\omega_k^{012}-\Delta_3\omega_k^{023})e_{02}^k\\ \nonumber
+(\Delta_1\omega_k^{013}+\Delta_2\omega_k^{023})e_{03}^k
+(\Delta_0\omega_k^{012}-\Delta_3\omega_k^{123})e_{12}^k\\
+(\Delta_0\omega_k^{013}+\Delta_2\omega_k^{123})e_{13}^k
+(\Delta_0\omega_k^{023}-\Delta_1\omega_k^{123})e_{23}^k\big],
\end{eqnarray}
\begin{eqnarray}\label{eq:13}
\delta^c\overset{4}{\omega}=\sum_k\big[(\Delta_3\overset{4}{\omega}_k)e_{012}^k-(\Delta_2\overset{4}{\omega}_k)e_{013}^k
+(\Delta_1\overset{4}{\omega}_k)e_{023}^k+(\Delta_0\overset{4}{\omega}_k)e_{123}^k\big].
\end{eqnarray}
Let $\Omega\in K(4)$ be a discrete inhomogeneous, that is
\begin{equation}\label{eq:14}
\Omega=\sum_{r=0}^4\overset{r}{\omega},
\end{equation}
where $\overset{r}{\omega}\in K^r(4)$ is given by (\ref{eq:04}) and (\ref{eq:05}). A discrete analogue of the Dirac-K\"{a}hler equation (\ref{eq:01}) can be defined as
 \begin{equation}\label{eq:15}
i(d^c+\delta^c)\Omega=m\Omega.
\end{equation}
We can write this equation more explicitly by separating its homogeneous components as
\begin{eqnarray}\label{eq:16} \nonumber
i\delta^c\overset{1}{\omega}=m\overset{0}{\omega}, \quad i(d^c\overset{1}{\omega}+\delta^c\overset{3}{\omega})=m\overset{2}{\omega}, \quad
id^c\overset{3}{\omega}=m\overset{4}{\omega},\\
i(d^c\overset{0}{\omega}+\delta^c\overset{2}{\omega})=m\overset{1}{\omega}, \qquad
i(d^c\overset{2}{\omega}+\delta^c\overset{4}{\omega})=m\overset{3}{\omega}.
\end{eqnarray}
Substituting (\ref{eq:07})--(\ref{eq:13})  into (\ref{eq:16}) one  obtains the set of 16 difference equations   \cite{S2}.

As in \cite{S3},  we define  the Clifford multiplication of the basis elements $x^k$ and $e^k_\mu$,  \  $\mu=0,1,2,3$, by the following rules:
\begin{align*}
&\mbox{(a)} \quad x^kx^k=x^k, \quad x^ke^k_\mu=e^k_\mu x^k=e^k_\mu,\\
&\mbox{(b)} \quad e^k_\mu e^k_\nu+e^k_\nu e^k_\mu=2g_{\mu\nu}x^k, \quad g_{\mu\nu}=\mbox{diag}(1,-1,-1,-1),\\
&\mbox{(c)} \quad e^k_{\mu_1}\cdots e^k_{\mu_s}=e^k_{\mu_1\cdots \mu_s} \quad \mbox{for} \quad 0\leq \mu_1<\cdots <\mu_s\leq 3,
\end{align*}
supposing the product to be zero in all other cases.

The operation is linearly extended to arbitrary discrete forms.

Consider the following unit forms
\begin{equation}\label{eq:17}
x=\sum_kx^k, \qquad e=\sum_ke^k, \qquad e_\mu=\sum_ke_\mu^k, \qquad e_{\mu\nu}=\sum_ke_{\mu\nu}^k,
\end{equation}
where $\mu, \nu=0,1,2,3$.
Note that the unit 0-form $x$ plays  a role of the unit element in $K(4)$, i.e., for any $r$-form  $\overset{r}{\omega}$ we have $x\overset{r}{\omega}=\overset{r}{\omega}x=\overset{r}{\omega}$.
\begin{proposition}\label{1}
For the unit forms $x\in K^0(4)$ and $e_\mu\in K^1(4)$ given by (\ref{eq:17}) the following holds
\begin{equation}\label{eq:18}
e_\mu e_\nu+e_\nu e_\mu=2g_{\mu\nu}x, \qquad \mu,\nu=0,1,2,3.
\end{equation}
\end{proposition}
\begin{proof}
\smartqed
By the rule (b), it is obvious.
\qed
\end{proof}
\begin{proposition}\label{2}
Let $\Omega\in K(4)$ be an inhomogeneous discrete form. Then we have
\begin{equation}\label{eq:19}
(d^c+\delta^c)\Omega=\sum_{\mu=0}^3e_\mu\Delta_\mu\Omega,
\end{equation}
where
 $\Delta_\mu$ is the difference operator which acts on each component of $\Omega$ by the rule~(\ref{eq:06}).
\end{proposition}
\begin{proof}
\smartqed
See Proposition~1 in  \cite{S4}.
\qed
\end{proof}
Thus the discrete Dirac-K\"{a}hler equation (\ref{eq:15}) can be rewritten in the form
\begin{equation*}
i\sum_{\mu=0}^3e_\mu\Delta_\mu\Omega=m\Omega.
\end{equation*}
Let  $K^{ev}(4)=K^0(4)\oplus K^2(4)\oplus K^4(4)$ and let $\Omega^{ev}\in K^{ev}(4)$ be a real-valued even inhomogeneous form, i.e., $\Omega^{ev}=\overset{0}{\omega}+\overset{2}{\omega}+\overset{4}{\omega}$.  A discrete analogue of the Hestenes equation (\ref{eq:02}) is defined by
\begin{equation}\label{eq:20}
-(d^c+\delta^c)\Omega^{ev} e_1e_2=m\Omega^{ev} e_0,
\end{equation}
or equivalently,
\begin{equation*}
-\sum_{\mu=0}^3e_\mu\Delta_\mu\Omega^{ev}e_1e_2=m\Omega^{ev} e_0,
\end{equation*}
where $e_1, e_2$ and $e_0$ are given by  (\ref{eq:17}).
A discrete analogue of the Joyce equation (\ref{eq:03}) is given by
\begin{equation}\label{eq:21}
i(d^c+\delta^c)\Omega^{ev}=m\Omega^{ev} e_0,
\end{equation}
where $\Omega^{ev}\in K^{ev}(4)$ is a complex-valued even inhomogeneous form.
Clearly, Eq.~(\ref{eq:21}) can be rewritten in the form
\begin{equation*}
i\sum_{\mu=0}^3e_\mu\Delta_\mu\Omega^{ev}=m\Omega^{ev} e_0.
\end{equation*}
Applying (\ref{eq:07})--(\ref{eq:13}) Eqs.~(\ref{eq:20}) and (\ref{eq:21})   can be expressed also in terms of difference equations (see \cite{S3, S4}).

\section{Chirality and a discrete Joyce equation}
\label{sec:3}
In the continuum Dirac theory, the fifth gamma matrix $\gamma_5$ defined by $\gamma_5=i\gamma_0\gamma_1\gamma_2\gamma_3$ plays a central role in formulating chiral fermions.
It is known that in the language of differential forms the Hodge star operator $\ast$ has similar properties, up to sign, as $\gamma_5$.
The difficulties in defining a discrete Hodge star operator to deal with chirality  on the lattice  were  discussed by Rabin in \cite{Rabin}.
Several discrete versions of the Hodge star operator have been proposed in \cite{Beauce,S2,Watterson} in which the chiral properties for Dirac-K\"{a}hler fermions in the geometric discretisation are captured. In this section, we use a discrete analogue of $\gamma_5$ to describe the chirality of a discrete Dirac field in the Joyce formulation.

Consider the constant 4-form $e_5$ defined by
\begin{equation}\label{eq:22}
e_5 = ie_0e_1e_2e_3=ie,
\end{equation}
where $e_\mu\in K^1(4)$ and $e\in K^4(4)$ are given by (\ref{eq:17}). The form $e_5$ generates the action 
$e_5:  \overset{r}{\omega}\rightarrow e_5\overset{r}{\omega}$,
where $\overset{r}{\omega}\in K^r(4)$.  Note also that
\begin{equation*}
e_5:  K^{r}(4)\rightarrow K^{4-r}(4).
\end{equation*}
It is easy to check that
\begin{equation}\label{eq:23}
e_5^2 = x \quad  \mbox{and} \quad  e_5e_\mu=-e_\mu e_5 \quad \mbox{for} \quad \mu=0,1,2,3.
\end{equation}
Hence  the form $e_5\in K^4(4)$ has exactly the same properties as $\gamma_5$.
\begin{proposition}\label{3}
For any inhomogeneous form $\Omega\in K(4)$ we have
\begin{equation}\label{eq:24}
e_5(d^c+\delta^c)\Omega=-(d^c+\delta^c)e_5\Omega.
\end{equation}
\end{proposition}
\begin{proof}
\smartqed
By Proposition~\ref{2} and (\ref{eq:23}), the equality (\ref{eq:24}) follows.
\qed
\end{proof}
Consider the following constant forms
\begin{equation}\label{eq:25}
P_L=\frac{x-e_5}{2}, \qquad  P_R=\frac{x+e_5}{2}.
\end{equation}
Since
\begin{equation*}
P_L^2=P_LP_L=P_L, \qquad  P_R^2=P_RP_R=P_R,
\end{equation*}
it follows that $P_L$ and $P_R$ are projectors. Let us represent $\Omega\in K(4)$ as
\begin{equation}\label{eq:26}
\Omega=\Omega_L+\Omega_R,
\end{equation}
where
\begin{equation}\label{eq:27}
\Omega_L=P_L\Omega, \qquad  \Omega_R=P_R\Omega.
\end{equation}
It is clear that $e_5\Omega_R=\Omega_R$ and $e_5\Omega_L=-\Omega_L$. Hence we can say that $\Omega$ decomposes into its self-dual and anti-self-dual parts with respect to the action $e_5$.
The self-dual and anti-self-dual components of $\Omega$  correspond to the chiral right and chiral left parts of a solution of the discrete Dirac-K\"{a}hler equation.
\begin{proposition}\label{4}
 If   $\Omega$ is a solution of the massless discrete Dirac-K\"{a}hler equation
 \begin{equation}\label{eq:28}
i(d^c+\delta^c)\Omega=0,
\end{equation}
  then so are both  $\Omega_R$ and $\Omega_L$.
 \end{proposition}
\begin{proof}
\smartqed
Let $\Omega\in K(4)$ be  a solution of Eq.(\ref{eq:28}). Using (\ref{eq:24}) and (\ref{eq:27}) we obtain
\begin{equation*}
i(d^c+\delta^c)(\Omega\pm e_5\Omega)=i(d^c+\delta^c)\Omega\mp e_5i(d^c+\delta^c)\Omega=0.
\end{equation*}
\qed
\end{proof}
From Proposition~\ref{4} it follows immediately that the massless  discrete Dirac-K\"{a}hler equation is invariant under the transformation
\begin{equation}\label{eq:29}
\Omega\longrightarrow\Omega\pm e_5\Omega.
\end{equation}
In other words, the discrete model admits the chiral symmetry (\ref{eq:29}) of Eq.~(\ref{eq:28}) with respect to the action $e_5$.
\begin{proposition}\label{5}
 If   $\Omega$ is a solution of the  discrete Dirac-K\"{a}hler equation (\ref{eq:15}) then we have
 \begin{eqnarray*} 
i(d^c+\delta^c)\Omega_L=m\Omega_R,  \\ i(d^c+\delta^c)\Omega_R=m\Omega_L. 
\end{eqnarray*}
 \end{proposition}
 \begin{proof}
 \smartqed
From (\ref{eq:24}) it follows that
\begin{equation}\label{eq:30}
(d^c+\delta^c)P_L\Omega=P_R(d^c+\delta^c)\Omega, \quad (d^c+\delta^c)P_R\Omega=P_L(d^c+\delta^c)\Omega
\end{equation}
for any $\Omega\in K(4)$.
Let $\Omega$ be a solution of Eq.~(\ref{eq:15}).
By (\ref{eq:30}), we have
\begin{equation*}
i(d^c+\delta^c)\Omega_L=i(d^c+\delta^c)P_L\Omega=P_R(m\Omega)=m\Omega_R
\end{equation*}
and
\begin{equation*}
i(d^c+\delta^c)\Omega_R=i(d^c+\delta^c)P_R\Omega=P_L(m\Omega)=m\Omega_L.
\end{equation*}
\qed
\end{proof}
Hence, just as in the continuum case,  the operator  $i(d^c+\delta^c)$ flips the chirality and the massive discrete Dirac-K\"{a}hler equation decomposes into two parts.

 Let $\Omega^{ev}\in K^{ev}(4)$ be a complex-valued even inhomogeneous form. Then we have
 \begin{equation*}
\Omega^{ev}=P_L\Omega^{ev}+P_R\Omega^{ev}=\Omega^{ev}_L+\Omega^{ev}_R,
\end{equation*}
where $P_L$ and $P_R$ are given by (\ref{eq:25}). The discrete Joyce equation splits into two parts in the following way.
\begin{proposition}\label{6}
 If   $\Omega^{ev}$ is a solution of the  discrete Joyce equation (\ref{eq:21}) then we have
  \begin{eqnarray}\label{eq:31}
i(d^c+\delta^c)\Omega^{ev}_L=m\Omega^{ev}_Re_0,  \\ i(d^c+\delta^c)\Omega^{ev}_R=m\Omega^{ev}_Le_0.\label{eq:32}
\end{eqnarray}
 \end{proposition}
 \begin{proof}
 \smartqed
The proof is the same as that for Proposition~\ref{5}.
\qed
\end{proof}
Thus the chiral properties are captured for our discrete model.

\section{Chirality and a discrete Hestenes equation}
\label{sec:4}
Recall that the Hestenes equation is a form of the Dirac equation in the real algebra $\emph{C}\ell_{\mathbb{R}}(1,3)$. The discrete Hestenes equation acts in the space of real-valued even form $K^{ev}(4)$.
Unfortunately, to discus the chiral properties of this equation the action (\ref{eq:22}) makes no sense because the form $e_5$ defined by (\ref{eq:22}) is complex-valued. To make sense of the chiral action one must substitute for $e_5$  a real-valued action.  Let us denote by $\ast_5$ the following transformation
\begin{equation}\label{eq:33}
\ast_5:  \overset{r}{\omega}\rightarrow e\overset{r}{\omega} e_2e_1,
\end{equation}
where $\overset{r}{\omega}\in K^r(4)$ and  $e, e_2, e_1$ are given by (\ref{eq:17}). It is true that $\ast_5:  K^{ev}(4)\rightarrow K^{ev}(4)$.

\begin{proposition}\label{7}
For any   inhomogeneous form $\Omega\in K(4)$ we have
\begin{equation}\label{eq:34}
(\ast_5)^2\Omega=\Omega, \quad \mbox{and} \quad (\ast_5e_\mu+e_\mu \ast_5)\Omega=0 \quad \mbox{for} \quad \mu=0,1,2,3.
\end{equation}
\end{proposition}
\begin{proof}
 \smartqed
By definition, $e=e_0e_1e_2e_3$   and  $e^2=ee=-x$.  Then for any $\overset{r}{\omega}\in K^r(4)$ we have
\begin{equation*}
(\ast_5)^2\overset{r}{\omega}=\ast_5(\ast_5\overset{r}{\omega})= e(e\overset{r}{\omega}e_2e_1)e_2e_1=x\overset{r}{\omega}x=\overset{r}{\omega}.
\end{equation*}
Since $e\in K^4(4)$ anticommutes with $e_\mu\in K^1(4)$ for $\mu=0,1,2,3$, i.e., $ee_\mu=-e_\mu e$,  the second equality of  (\ref{eq:34}) follows immediately.
\qed
\end{proof}

\begin{proposition}\label{8}
Let $\Omega\in K(4)$ be an inhomogeneous form. Then  the following holds
\begin{equation}\label{eq:35}
(\ast_5(d^c+\delta^c)+(d^c+\delta^c)\ast_5)\Omega=0.
\end{equation}
\end{proposition}
\begin{proof}
 \smartqed
By  (\ref{eq:34}),  the proof repeats the proof of  Proposition~\ref{3}.
\qed
\end{proof}
From (\ref{eq:34}) and (\ref{eq:35}) it follows that to deal with chirality in the case of the discrete Hestenes equation   one can take $\ast_5$.
\begin{proposition}\label{9}
The massless  discrete Dirac-K\"{a}hler equation is invariant under the transformation
\begin{equation}\label{eq:36}
\Omega\longrightarrow\Omega\pm \ast_5\Omega.
\end{equation}
\end{proposition}
\begin{proof}
 \smartqed
By (\ref{eq:35}),  it is obvious.
\qed
\end{proof}
It follows that the discrete model admits a chiral symmetry of the type  (\ref{eq:36}).

Let us consider the following operations
\begin{equation}\label{eq:37}
P^*_L=\frac{1-\ast_5}{2},\qquad P^*_R=\frac{1+\ast_5}{2}.
\end{equation}
It is easy to check that
\begin{equation*}
(P^*_L)^2\Omega=P^*_L\Omega,\qquad (P^*_R)^2\Omega=P^*_R\Omega, \qquad P^*_LP^*_R\Omega=P^*_RP^*_L\Omega=0
\end{equation*}
for any $\Omega\in K(4)$.
Hence, the operations  $P^*_L$ and $P^*_R$  are projectors. Then $\Omega\in K(4)$  can be represented as (\ref{eq:26}),
where
\begin{equation*}
\Omega_L=P^*_L\Omega, \qquad  \Omega_R=P^*_R\Omega.
\end{equation*}
Let $\Omega^{ev}\in K^{ev}(4)$ be a real-valued even inhomogeneous form. Then the forms
$\Omega^{ev}_L=P^*_L\Omega^{ev}$ and  $\Omega^{ev}_R=P^*_R\Omega^{ev}$ are even and we have
\begin{equation*}
\Omega^{ev}=\Omega^{ev}_L+\Omega^{ev}_R.
\end{equation*}
It should be noted that  $\Omega^{ev}_R$ and $\Omega^{ev}_L$ are self-dual and  anti-self-dual parts of $\Omega^{ev}$  with respect to the action $\ast_5$. They correspond to the chiral right and chiral left parts of a solution of the discrete Hestenes equation. Similarly, as in the case of the Joyce equation, we have the following decomposition of the discrete Hestenes equation.
\begin{proposition}\label{10}
 If   $\Omega^{ev}$ is a solution of the  discrete Hestenes equation (\ref{eq:20}) then we have
  \begin{eqnarray*}
-(d^c+\delta^c)\Omega^{ev}_Le_1e_2=m\Omega^{ev}_Re_0,  \\ -(d^c+\delta^c)\Omega^{ev}_Re_1e_2=m\Omega^{ev}_Le_0.
\end{eqnarray*}
 \end{proposition}
 \begin{proof}
 \smartqed
Using  (\ref{eq:35}) and (\ref{eq:37}) we obtain
\begin{equation*}
(d^c+\delta^c)P^*_L\Omega=P^*_R(d^c+\delta^c)\Omega, \quad (d^c+\delta^c)P^*_R\Omega=P^*_L(d^c+\delta^c)\Omega
\end{equation*}
for any $\Omega\in K(4)$.
Therefore the proof repeats the proof of  Proposition~\ref{5}.
\qed
\end{proof}
 Let us consider the parity operation   $P: K^r(4)\rightarrow K^r(4)$ defined by
\begin{equation}\label{eq:38}
P\overset{r}{\omega}=e_0\overset{r}{\omega} e_0,
\end{equation}
where $\overset{r}{\omega}\in K^r(4)$ and $e_0\in K^1(4)$ is given by (\ref{eq:17}). It is clear that $P^2\overset{r}{\omega}=\overset{r}{\omega}$. But the second statement of Proposition~\ref{7} is not true.
The parity operation (\ref{eq:38}) changes the chirality  of discrete forms in the following way.
\begin{proposition}\label{11}
  For any form $\Omega\in K(4)$ we have
  \begin{equation}\label{eq:39}
P(P^*_L\Omega)=P^*_R(P\Omega),  \qquad  P(P^*_R\Omega)=P^*_L(P\Omega),
\end{equation}
where $P^*_L$ and $P^*_R$ are given by (\ref{eq:37}).
 \end{proposition}
 \begin{proof}
 \smartqed
 Since $e_0$ commutes with $e_2e_1$ and anticommutes with $e$ it follows immediately.
\qed
\end{proof}
Decompose an even inhomogeneous form  $\Omega^{ev}\in K(4)$ as follows
\begin{equation*}
\Omega^{ev}=\Omega^{ev}_++\Omega^{ev}_-,
\end{equation*}
 where $\Omega^{ev}_+$ commutes with $e_0$ and   $\Omega^{ev}_-$ anticommutes with it, i.e.,
 \begin{equation}\label{eq:40}
e_0\Omega^{ev}_{\pm}=\pm\Omega^{ev}_{\pm}e_0.
\end{equation}
\begin{proposition}\label{12}
  Let  $\Omega^{ev}_{\pm R}=P^*_R\Omega^{ev}_{\pm}$ and   $\Omega^{ev}_{\pm L}=P^*_L\Omega^{ev}_{\pm}$. Then we have
  \begin{eqnarray*}
P\Omega^{ev}_{+ R}=\Omega^{ev}_{+L},  \qquad  P\Omega^{ev}_{- R}=-\Omega^{ev}_{-L}, \\
P\Omega^{ev}_{+L}=\Omega^{ev}_{+R},  \qquad  P\Omega^{ev}_{-L}=-\Omega^{ev}_{-R}.
\end{eqnarray*}
 \end{proposition}
\begin{proof}
 \smartqed
By (\ref{eq:38})--(\ref{eq:40}),  we obtain
\begin{equation*}
P\Omega^{ev}_{+ R}=P(P^*_R\Omega^{ev}_{+})=P^*_L(P\Omega^{ev}_{+})=P^*_L(e_0\Omega^{ev}_{+}e_0)=P^*_L(\Omega^{ev}_{+}e_0e_0)=P^*_L\Omega^{ev}_{+}=\Omega^{ev}_{+L}.
\end{equation*}
The same proof remains valid for all other cases.
\qed
\end{proof}

\section{Discrete plane wave solutions and spin eigenstates}
\label{sec:4}

Discrete versions of the plane wave solutions to discrete Joyce and Hestenes equations are constructed in \cite{S5}  and \cite{S6}.
In this section, we study spin properties of these solutions in the case of the discrete Joyce equation and discus how the chirality is realized for spin eigenstates in our discrete model.
Recall a discrete version of the  general  plane wave solution  for the Joyce equation (see for details \cite{S5}). Let $\psi\in K^0(4)$ and let
\begin{equation*}\label{eq:}
 \psi=\sum_k(ip_0+1)^{k_0}(ip_1+1)^{k_1}(ip_2+1)^{k_2}(ip_3+1)^{k_3}x^k,
 \end{equation*}
 where $i$ is the usual complex unit, $p_\mu\in\mathbb{R}$  and $k=(k_0,k_1,k_2,k_3)$ is a multi-index.
 Let  $A$ be the even inhomogeneous form given by
 \begin{eqnarray*}\label{}
A=a_1((m-p_0)x+p_1e_{01}+p_2e_{02}+p_3e_{03})\nonumber\\+a_2((m-p_0)e_{12}+p_2e_{01}-p_1e_{02}+p_3e)\nonumber \\
  +a_3((m-p_0)e_{13}+p_3e_{01}-p_1e_{03}+p_2e)\nonumber\\
    +a_4((m-p_0)e_{23}+p_3e_{02}-p_2e_{03}+p_1e),
 \end{eqnarray*}
 where
    $p_0=\pm\sqrt{m^2+p_1^2+p_2^2+p_3^2}$, \  $a_\mu=\frac{\alpha_\mu}{m-p_0}$ \ and $\alpha_\mu$ is an arbitrary complex number 
    for \ $\mu=1,2,3,4$. \
  Here the even unit forms \ $x\in K^0(4)$, \ $e\in K^4(4)$ \ and \ $e_{\mu\nu}\in K^2(4)$ are given by (\ref{eq:17}).
Then the most general plane wave solution of Eq.~(\ref{eq:21}) is
 \begin{equation}\label{eq:41}
  \Omega^{ev}=A\psi.
 \end{equation}
  Let consider a  particular case of (\ref{eq:41}), namely  $p_2=p_3=0$. This situation corresponds to one in the continuum case in which the plane wave solution is propagating along only one axis, e.g., $x_1$.
 In the continuum case, such solutions for a Dirac generalized bivector equation are described in \cite{J}.
Then we have
\begin{equation}\label{eq:42}
 \psi=\sum_k(ip_0+1)^{k_0}(ip_1+1)^{k_1}x^k
 \end{equation}
and
\begin{eqnarray}\label{eq:43}
A=a_1((m-p_0)x+p_1e_{01})+a_2((m-p_0)e_{12}-p_1e_{02})\nonumber \\
  +a_3((m-p_0)e_{13}-p_1e_{03})+a_4((m-p_0)e_{23}+p_1e).
 \end{eqnarray}
  Let us introduce the following constant 2-forms
\begin{equation}\label{eq:44}
   S_1=i\frac{1}{2}e_{23}, \qquad S_2=-i\frac{1}{2}e_{13}, \qquad S_3=i\frac{1}{2}e_{12}.
 \end{equation}
 By definition, we have  $e_{12}e_{12}=e_{13}e_{13}=e_{23}e_{23}=-x$  and one may easy calculate that
  $S_1^2+S_2^2+S_3^2=\frac{1}{2}(\frac{1}{2}+1)x$.
Hence similarly to the continuum case the forms (\ref{eq:44}) can be interpreted  as  spin  operators for our discrete model and  spin eigenstates of $\pm\frac{1}{2}$ along the direction of propagation can be described for the solution (\ref{eq:41}), where $\psi$ and $A$ are given by (\ref{eq:42}) and (\ref{eq:43}).

An easy computation shows that the equations
   $S_2A=\frac{1}{2}A$  and  $S_3A=\frac{1}{2}A$,
 where $A$ is given by (\ref{eq:43}), have only trivial solutions, i.e., $a_1=a_2=a_3=a_4=0$.
 However, the equation $S_1A=\frac{1}{2}A$ has an non-trivial solution.
 Indeed, applying the spin operator $S_1$ to (\ref{eq:43}) we obtain
 \begin{eqnarray}\label{eq:45}
S_1A=i\frac{1}{2}\big(a_1(m-p_0)e_{23}+a_1p_1e+a_2(m-p_0)e_{13}-a_2p_1e_{03}\nonumber \\
  -a_3(m-p_0)e_{12}+a_3p_1e_{02}-a_4(m-p_0)x-a_4p_1e_{01}\big).
 \end{eqnarray}
Combining (\ref{eq:45}) with  (\ref{eq:43}) we conclude that $S_1A=\frac{1}{2}A$ if and only if $a_1=-ia_4$ and $a_2=-ia_3$.
It follows that $A$ can be represented as
\begin{equation}\label{eq:46}
   A=a_1A_1+a_2A_2,
 \end{equation}
where
\begin{eqnarray}\label{eq:47}
A_1=(m-p_0)x+p_1e_{01}+i(m-p_0)e_{23}+ip_1e,\nonumber \\
 A_2=(m-p_0)e_{12}-p_1e_{02}+i(m-p_0)e_{13}-ip_1e_{03},
 \end{eqnarray}
 and $a_1, a_2$ are arbitrary constant. Since $\psi$ is a 0-form we have
 $S_1\Omega^{ev}=\frac{1}{2}\Omega^{ev}$, where $\Omega^{ev}$ is the plane wave solution (\ref{eq:41}) and $A$ is given by (\ref{eq:46}). On  other words,  $\Omega^{ev}$  is an eigenstate corresponding to the eigenvalue $\frac{1}{2}$ of the spin operator $S_1$.

It is clear that if $A\psi$ is a solution of the discrete Joyce equation then $\bar{A}\psi$, where $\bar{A}$ denotes the complex conjugate of $A$,  is also a solution.
It can also be seen that $\bar{A}=\bar{a}_1\bar{A}_1+\bar{a}_2\bar{A}_2$ satisfies the equation $S_1\bar{A}=-\frac{1}{2}\bar{A}$.
Hence, similarly as in continuum case \cite{J} the solutions $A\psi$ and $\bar{A}\psi$, where  $\psi$ and  $A$ are given by (\ref{eq:42}) and (\ref{eq:46}), can be interpreted as spin up and spin down solutions correspondingly.

It should be noted that the chirality is captured for the spin solutions described above.
Applying the projectors (\ref{eq:25}) to the forms $A_1$ and $A_2$ given by (\ref{eq:47}) one can calculate
\begin{eqnarray*}\label{}
P_RA_1=\frac{1}{2}(m-p_0+p_1)(x+e_{01}+ie_{23}+ie),\nonumber \\
P_LA_1=\frac{1}{2}(m-p_0-p_1)(x-e_{01}+ie_{23}-ie),\nonumber \\
 P_RA_2=\frac{1}{2}(m-p_0+p_1)(e_{12}-e_{02}+ie_{13}-ie_{03}),\nonumber \\
 P_LA_2=\frac{1}{2}(m-p_0-p_1)(e_{12}+e_{02}+ie_{13}+ie_{03}).\nonumber
 \end{eqnarray*}
Thus we have the following two left and two right chiral states
\begin{equation}\label{eq:48}
 \Omega^{ev}_{1L}=P_LA_1\psi, \quad \Omega^{ev}_{2L}=P_LA_2\psi, \quad \Omega^{ev}_{1R}=P_RA_1\psi, \quad \Omega^{ev}_{2R}=P_RA_2\psi,
 \end{equation}
 where $\psi$ is given by (\ref{eq:41}). Obviously, as has already been described in Sect.~\ref{sec:3} the forms (\ref{eq:48}) satisfy Eqs.~(\ref{eq:31}) and (\ref{eq:32}).

\end{document}